\documentclass{article}

\usepackage{amsmath}
\usepackage{amsthm}
\usepackage{amssymb}
\usepackage{amsrefs}
\usepackage{algorithm}
\usepackage[noend]{algpseudocode}
\usepackage{graphicx}
\usepackage[font=small]{caption}
\usepackage{float}
\usepackage{authblk}
\usepackage{bbm}

\newtheorem{thm}{Theorem}[section]
\newtheorem{prop}{Proposition}

\newtheorem{rem}{Remark}[thm]

\let\[\equation
\let\]\endequation

\begin{document}

\title{On Unbiased Simulations of Stochastic Bridges Conditioned on Extrema}
\date{November, 2019}
\author[]{Andrew Schaug\thanks{andrew.t.schaug@ey.com} \, and Harish Chandra\thanks{harish.chandra@ey.com}}
\affil[]{Quantitative Advisory Services, Ernst \& Young LLP}


\maketitle

\begin{abstract}
Stochastic bridges are commonly used to impute missing data with a lower sampling rate to generate data with a higher sampling rate, while preserving key properties of the dynamics involved in an unbiased way. While the generation of Brownian bridges and Ornstein-Uhlenbeck bridges is well understood, unbiased generation of such stochastic bridges subject to a given extremum has been less explored in the literature. After a review of known results, we compare two algorithms for generating Brown-ian bridges constrained to a given extremum, one of which generalizes to other diffusions. We further apply this to generate unbiased Ornstein-Uhlenbeck bridges and unconstrained processes, both constrained to a given extremum, along with more tractable numerical approximations of these algorithms. Finally, we consider the case of drift and applications to geometric Brownian motions.
\end{abstract}

\section{Background}

Stochastic bridges find applications in a wide range of domains from finance to climatology. It is often necessary to interpolate insufficiently frequent time series data stochastically while preserving as much as possible of the dynamics underlying the original data. A natural question then is how to construct a stochastic bridge, i.e., a stochastic process whose end points are known a priori. These have been used for studying spectral statistics of polarization mode dispersion \cite{S2}, modeling fermions trapped in a trap \cite{DMS}, modeling animal movement patterns \cite{BSAW}, valuing financial securities \cite{CMO} and improving the performance of quasi-Monte Carlo methods for high dimensional problems in computational finance \cite{LW} \cite{K}. Sometimes, there is a further need to generate a bridge subject to additional constraints, such as a given mean (see \cite {M}), or on an extremum. For example, a scenario generator may wish to consider worst-case scenarios where an asset attains a given lowest value to test how a financial instrument would behave under certain unfavorable market conditions, or we may wish to interpolate simulated thermodynamic scenarios where the temperature reaches a certain maximum.

This paper will provide tools for industry to go beyond simple Brownian bridges to allow for some features of more realistic dynamics including drift, geometric Brownian motion and mean reversion, for use in stress tests that condition on extreme values of an asset. For example, the generation of geometric Brownian bridges subject to extrema may be used for the valuation of barrier options whose barrier event has not yet been breached. The algorithms presented in this paper could also be used by banks for data quality remediation problems in VaR models where input time series data is sparse. 

We first show how to generate Wiener processes subject to bridge and/or extremum conditions in an unbiased (i.e., measure-preserving) way, and then provide and compare two methods for generating such processes subject to both bridge and extremum conditions. Our first method is much more rapid, while the second generalizes more broadly. We will then consider mean-reverting dynamics by extending this problem to Ornstein-Uhlenbeck processes. Finally, we consider Wiener processes with drift, and geometric Brownian motions, which underlie the Black-Scholes pricing model.

We now introduce some preliminaries. First, we slightly abuse terminology by allowing a Wiener process to start at an arbitrary initial value $a$. A \textit{Brownian bridge} is a Wiener process $X_t$ over $[0, T]$ conditioned on $X_0 = a, X_T = b$. An Ornstein-Uhlenbeck (OU) process is a simple mean-reverting stochastic process $X^{OU}_t$given by the SDE \begin{equation}\label{ou_sde}dX^{OU}_t = \kappa(\mu - X_t)dt + \sigma dW_t,\end{equation} for some standard Wiener process $W_t$. Then $\mu$ is the mean and $\kappa$ is the \textit{mean reversion rate}. An \textit{Ornstein-Uhlenbeck bridge} is an OU process over a closed interval conditioned on the values at the endpoints.

In section (2) we provide two unbiased constructions of Wiener bridges conditioned on extrema. In section (3) we generalize the second of the methods in section (2) to construct OU bridges conditioned on an extremum, and discuss the numerical and computational concerns and reasonable approximations involved. In section (4) we consider how to construct an \textit{open-ended} (i.e., non-bridge) Ornstein-Uhlenbeck process constrained on an extremum. In section (5) we consider the case of Wiener processes with drift, or equivalently geometric Brownian motions.

First, we discuss those problems already addressed in the literature (1-5), and then those which we address in this paper (6-9).

\begin{enumerate}
\item For the sake of completeness, a Wiener process without bridge or extremum conditions has a trivial construction. Computationally, the Euler-Maruyama method may be used, successively incrementing by $\sigma \sqrt{\Delta t}\xi$, where $\xi \sim \mathcal{N}(0, 1)$.
\item An OU process without bridge or extremum conditions may be constructed from a Wiener process $W_t$ directly via equation \ref{ou_sde}. Computationally, the Euler-Maruyama method may be used.
\item A Brownian bridge may be constructed as follows: Given a Wiener process $W_t$ with variance $\sigma^2$ and such that $W_0 = 0$, we set $$X_t = a + \frac{b-a}{T}t + W_t - \frac{W_T}{T} t.$$
Then $X_t$ is a Brownian bridge on $[0, T]$ from $a$ to $b$, also with variance $\sigma^2$. Furthermore, the induced mapping $W_t \mapsto X_t$ has the correct measure, that is, the construction of $X_t$ is unbiased provided the simulation of $W_t$ is unbiased. Equivalently, the SDE for such a Brownian bridge is given by \begin{equation}X_t = \frac{b-a-X_t}{T}dt + dW_t,\label{sde_bb}\end{equation} with initial condition $X_0 = a$.
\item An OU bridge may be constructed as follows: Given a Brownian bridge $W_t$ with variance $\sigma^2$ over $[0, T]$ such that $W_0 = 0, W_T = b-a$, set $X_0 = a-\mu$ and simulate the SDE \begin{equation} \label{ou_bridge_sde} dX_t = - \kappa X_t dt + \frac{2\kappa(be^{\kappa(t+T)}-X_te^{2\kappa t})}{e^{2\kappa T}-e^{2\kappa t}}dt + dW_t.\end{equation} Then $\mu + X_t$ is an OU bridge with variance $\sigma^2$ between $a$ and $b$. See \cite{C} for details.

\item An open-ended (i.e. non-bridge) Wiener process conditioned on an extremum may be constructed as follows. To simulate a Wiener process $X_t$ over an interval $[t_0,T]$ starting at $X_{t_0}=a$ and conditioned on a maximum $M$, we follow the solution for the equivalent dual case (i.e., for the minimum) in \cite{BPR}, following results in \cite{BCP}. First, assume without loss of generality that $a=0$. Construct a Brownian bridge $W_{br}^M (t)$ from $0$ to $M$ over the same interval $[t_0,T]$. Then, find $t_M$, the \textit{first} time this bridge hits $M$. Take the first portion of this process up to $t_M$, and append potentially repeated reflections as follows:
$$W_{0, max-refl}^M(t) = \left\{
	\begin{array}{ll}
		W_{br}^M(t)  & \mbox{if } 0 \le t \le t_{\mathrm{max}} \\
		2(\mathrm{min}_{t \le s \le t_M} W_{br}^M(s)) - W_{br}^M(t) & \mbox{if } t_{\mathrm{max}} \le t \le T
	\end{array}
\right.$$

This selects $W_{0,max-refl}^M$ as an unbiased Brownian motion starting at 0 with maximum $M$ on the same interval (and no condition at $T$). To construct this for an arbitrary initial point $a$, simply construct $$W_{a, max-refl}^M(t) = a + W_{0, max-refl}^{M-a}(t).$$

\item We discuss an algorithm that constructs an open-ended OU process conditioned on an extremum in section (4).
\item We demonstrate how to construct a Brownian bridge conditioned on an extremum in section (2), in two different ways, and compare the results.
\item We discuss an algorithm that constructs an OU bridge conditioned on an extremum in section (3).
\item We discuss an algorithm that constructs an open-ended Brownian process with positive drift (or equivalently a growing geometric Brownian motion) constrained on an extremum in section (5). The case  for bridges with drift follows immediately from case 7 above.
\end{enumerate}

\section{Wiener Bridges Conditioned on Extrema}

Without loss of generality, consider the space of Brownian bridges $W_t$ on $[0, T]$, with variance $\sigma^2$, $W_0 = a, W_1 = b$, conditioned on a \textit{maximum} $M$. We construct a generator for such processes, unbiased according to the standard measure (induced from the Wiener measure), in two different ways.

\subsection{Method 1: Construction via Brownian Meanders}

By \cite{I}, a `standard' Wiener meander $W^{me, r}_t$ over $[0, 1]$ with $W^{me, r}_0 = 0$ and $W^{me, r}_1 = r$ may be constructed from three standard Wiener bridges $W_t^{1, 2, 3}$ over $[0, 1]$ from 0 to 0 as follows:
\begin{equation}
\label{meander_eqn}
W^{me, r}_t = \sqrt{(rt + W_t^1)^2 + (W_t^2)^2 + (W_t^3)^2}.
\end{equation}
Note that this is a 3-dimensional Bessel process with drift. Furthermore, this construction is unbiased according to the standard measure.

\begin{rem}
\label{meander_rem}
Note that with the notation above, from the scaling property of Wiener processes and checking the values at the endpoints, it follows that $$W^{me, T, \sigma, (a, b)}_t = a + \sigma\sqrt{T}W^{me, r/(\sigma\sqrt{T})}_t$$ is a Wiener process with variance $\sigma^2$, and values $a$ at $t=0$ and $r$ at $t=T$, with the measure of a conditional Wiener process preserved.
\end{rem}

By \cite{AC}, we note that the joint density of the minimum and location of the minimum of a Brownian bridge with variance 1 from $a$ to $b$ over $[0, T]$ is
\begin{equation*}
\resizebox{.9\hsize}{!}{$p(\mathrm{min}X_t = m, \mathrm{argmin}X_t = \theta) = \frac{(a-m)(b-m)\sqrt{2T}}{\sqrt{\pi \theta^3 (T-\theta)^3}}e^{\frac{(a-b)^2}{2T} - \frac{(a-m)^2}{2\theta}-\frac{(b-m)^2}{2(T-\theta)}}\mathbbm{1}_{\{m<a, m<b\}},$}
\end{equation*}
and the density of the minimum alone is
\begin{equation*}
\frac{2}{T} (a+b-2m)e^{\frac{-2(a-m)(b-m)}{T}}\mathbbm{1}_{\{m<a, m<b\}}
\end{equation*}
so that the conditional density is, for $m<a, m<b$,
\begin{align}
& p(\mathrm{argmin}_{s\in[0, T]} X_s= \theta \vert \mathrm{min}_{s\in[0, T]} X_s = M) \nonumber\\
& \qquad = \frac{1}{\sqrt{2\pi}}\frac{(a-m)(b-m)}{(a+b-2m)}\left(\frac{T}{\theta(T-\theta)}\right)^{3/2}e^{\frac{(b-a)^2}{2T}  + \frac{2(a-m)(b-m)}{T} - \frac{(a-m)^2}{2\theta} - \frac{(b-m)^2}{2(T-\theta)}}. \label{density_argmax_bb}
\end{align}

This only depends on $a, b, m$ through $a-m, b-m, a-b$. Wiener processes with initial value $a$ and standard variation $\sigma$ transform positively and affinely to standard Wiener processes via $x \mapsto x + \frac{x-a}{\sigma}$, so this transformation preserves location of maximum and minimum. This transformation scales all random increments by $\sigma$, so is an isomorphism at the level of measure, and the probability density for the location of the maximum and minimum is preserved.

The minimum of a standard Brownian process is exactly the maximum of its negative, so negating $a, b, m$ in the above formula and setting $M = \mathrm{max}(M_t)$, we find after scaling by $\sigma$ that, for $M>a, M>b$,
\begin{align}
p(\mathrm{argmax}X_t = \theta \vert \mathrm{max} X_t = M) & =  \frac{1}{\sqrt{2\pi}\sigma}\frac{(M-a)(M-b)}{(2M-a-b)}\left(\frac{T}{\theta(T-\theta)}\right)^{3/2} \nonumber\\
\label{density_argmax_wiener} & \times e^{\frac{1}{\sigma^2}(\frac{(b-a)^2}{2T}  + \frac{2(M-a)(M-b)}{T} - \frac{(M-a)^2}{2\theta} - \frac{(M-b)^2}{2(T-\theta)})}.
\end{align}





This presents the following algorithm:

\begin{algorithm}[H]
\caption{Generates an unbiased Wiener process $X_t$ with variance $\sigma^2$ over $[0, T]$ conditioned on maximum $M$, with $X_0 = a, X_T = b$}
\label{alg:meanders}
\textbf{Input.} $\sigma, M, T, a, b$
\begin{algorithmic}
\Procedure{GenerateWienerBridge}{$\sigma$, $M$, $T$, $a$, $b$}
\State $\theta \leftarrow$ a time randomly generated according to equation \ref{density_argmax_bb}
\State $W_t^{\mathrm{left}, i}$ for $i = 1, 2, 3 \leftarrow$ standard Wiener processes
\State $W^{me, \theta, \sigma, (0, M-a)}_t \leftarrow$ Wiener meander generated from $W_t^{\mathrm{left}, i}$ 
\State $X^{left}_t \leftarrow M - W^{me, \theta, \sigma, (0, M-a)}_{T-t}$ \Comment Brownian meander over $[0, \theta]$ from \State \qquad\qquad\qquad\qquad\qquad\qquad\qquad\qquad\;\; $a$ to $M$ 
\State $W_t^{\mathrm{right}, i}$ for $i = 1, 2, 3 \leftarrow$ standard Wiener processes
\State $W^{me, T-\theta, \sigma, (0, M-b)}_{t} \leftarrow$ Wiener meander generated from $W_t^{\mathrm{right}, i}$
\State $X^{right}_t \leftarrow M - W^{me, T-\theta, \sigma, (0, M-b)}_{t}$ \Comment Brownian meander over $[\theta, T]$ from  \State \qquad\qquad\qquad\qquad\qquad\qquad\qquad\qquad\;\; $M$ to $b$
\State $X_t = $ concatenate($X^{left}_t, X^{right}_t$)
\EndProcedure
\end{algorithmic}
\textbf{Output}: $X_t$
\end{algorithm}


\begin{prop}
Steps 1-5 above generate an unbiased Brownian bridge over $[0, T]$ from $a$ to $b$ with variance $\sigma^2$ and maximum $M$.
\end{prop}
\begin{proof}
Since Wiener processes have stationary independent increments by definition, the selections of the location of the maximum and unscaled meanders are independent, so that in the Wiener measure we have (suppressing the conditions $W_0 = a, W_T = b$): \begin{multline*}p(W_t \vert \mathrm{max}_{s \in [0, T]} W_s = M)  = p(\mathrm{argmax}_{s \in [0, T]} W_s = \theta \vert \mathrm{max}_{s \in [0, T]} W_s = M)\\ \times p(W_t|_{[0, \theta]} \vert \mathrm{argmax}_{s \in [0, \theta]}W_s = \theta)p(W_t|_{[\theta, T]} \vert \mathrm{argmax}_{s \in [\theta, T]}W_s = \theta).\end{multline*} Each factor was generated according to its law. The scaling preserves the independence and ensures Brownian meanders of the correct variance, maximum and endpoints, so the resulting process is an unbiased Brownian bridge.
\end{proof}

We illustrate a few examples from a Python implementation of \ref{alg:meanders} below. First we generate Wiener bridges on $[0, 1]$ from $a = 1$ to $b = 3$ conditioned on a maximum of 5, with $\sigma = 1$:
\begin{figure}[H]
\centering
\includegraphics[width=0.4\textwidth, angle=0]{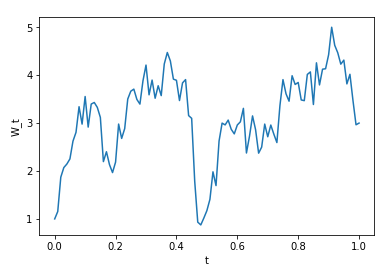} \includegraphics[width=0.4\textwidth, angle=0]{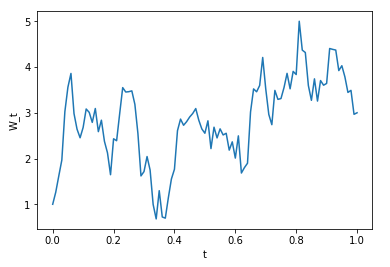}\\
\caption{Plots of Wiener bridges generated on $[0, 1]$ from $a = 1$ to $b = 3$ conditioned on a maximum of 5, with $\sigma = 1$}
\end{figure}

We then run this algorithm to generate Brownian bridges on $[0, 1]$ from $a = 1$ to $b = 3$ conditioned on a maximum of 5, with $\sigma = 5$:
 \begin{figure}[H]
\centering
\includegraphics[width=0.4\textwidth, angle=0]{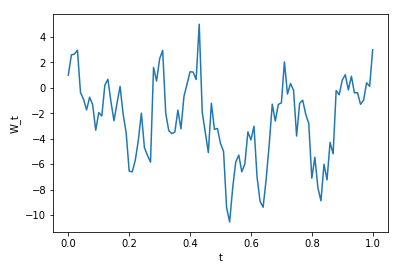} \includegraphics[width=0.4\textwidth, angle=0]{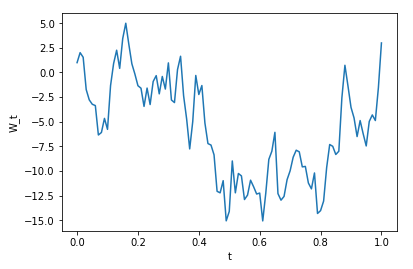}\\
\caption{Plots of Wiener bridges generated on $[0, 1]$ from $a = 1$ to $b = 3$ conditioned on a maximum of 5, with $\sigma = 5$}
\end{figure}

 As a final example, we run the algorithm to generate Brownian bridges on $[0, 1]$ from $a = 1$ to $b = 3$ conditioned on a maximum of 30, with $\sigma = 10$:
\begin{figure}[H]
\centering
\includegraphics[width=0.4\textwidth, angle=0]{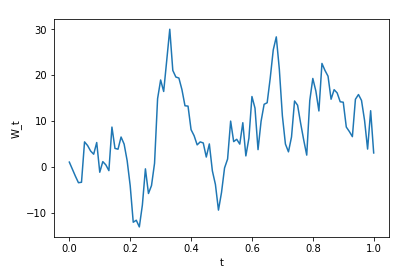} \includegraphics[width=0.4\textwidth, angle=0]{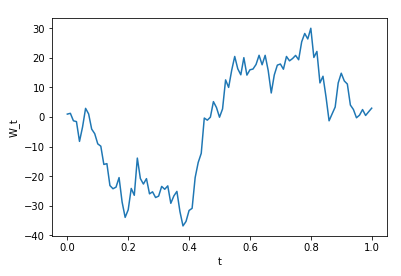}\\
\caption{Plots of Wiener bridges generated on $[0, 1]$ from $a = 1$ to $b = 3$ conditioned on a maximum of 30, with $\sigma = 10$. The maximum is considered attained within a tolerance of $\varepsilon = 0.1$}
\end{figure}

As expected, the process is forced far below the initial and final values, having high volatility but being constrained from above (but not below).

\subsection{Method 2: Construction via Bayesian Increments}

\subsubsection{Procedure}

We consider the full joint density of the maximum and location of the maximum as provided in \cite{AC}, to simulate the Brownian motion $X_t$ from $X_0=a$ with volatility $\sigma$ conditioned on $X_T=b$ and $\mathrm{max}_{t \in [t_0, T]} X_t = M$, without bias.

First, consider $t\in[t_0,T]$ and assume the standard measure. We start with $t = 0$, $X_t = a$, and successively increment $t$ by $\Delta t$. As we proceed there are two cases:

\begin{enumerate}
\item \textbf{The process has not yet achieved its maximum.} That is, we assume $X_s<M$ for all $s\in[t_0,t]$. We consider a discrete time increment $\Delta t$ and wish to find the probability density of $\Delta X$. By Bayes’ Law we have, density-wise,
\begin{align}
& p(X_{t+\Delta t}=X + \Delta X \,\vert\, \mathrm{max}_{s \in [t + \Delta t, T]} X_s = M, X_t = b) \nonumber\\
= & p(X_{t+\Delta t} = X + \Delta X \,\vert\, X_t = b )\nonumber\\
\label{bayes_wiener}& \times \frac{p(\mathrm{max}_{s \in [t, T]} X_s = M \,\vert\, X_{t+\Delta t} = X + \Delta X, X_t = b))}{p(\mathrm{max}_{s\in [t, T]} X_s = M \,\vert\ X_t = b)}
\end{align}

Here the condition $X_{t_0}=a$ pertains to every factor but may be suppressed since $X$ has independent and stationary increments.

We find
$$p(X_{t+\Delta t} = X + \Delta X \,\vert\, X_t = b) = \frac{1}{\sqrt{2\pi\sigma^2\Delta t}} e^{-(\frac{\Delta X - (b-X)}{T-t} \Delta t)^2/(2\sigma^2 \Delta t)},$$
directly from the SDE for an unconditioned Brownian bridge (equation \ref{sde_bb}).

For the remaining two factors, given the density of the minimum of a Brownian bridge as found in \cite{AC}, the symmetry of Wiener processes allows us to find the density of the maximum of a Brownian bridge by negating $a,b$ and $y$:
\begin{align}
& p(\mathrm{max}_{s\in [t_0, T]}X_s = M \, \vert \, X(t_0) = a, X_t = b)) \nonumber \\
= & \frac{2}{T-t_0} (2M-a-b) e^{\frac{-2(a-M)(b-M))}{T-t_0}} \mathbbm{1}_{\{ M \ge a \}} \mathbbm{1}_{\{M\ge b\}}. \label{density_max_bb}
\end{align}

Briefly suppressing the condition $(X_{t+\Delta t} = X + \Delta X, X_t = b)$, we consider whether the (almost surely unique) maximum is attained between $[t, t+\Delta t]$ or $[t+\Delta t, T]$, and further note that
\begin{align*}
 & p(\mathrm{max}_{s\in [t, T]}X_s = M) \\
 & = p(\mathrm{max}_{[t, t+\Delta t]} X_t = M \vert X_{t+\Delta t} = X_t + \Delta X_t) P(\mathrm{max}_{[t+\Delta t, T]} X_t \le M \vert X_T = b) \\
& + P(\mathrm{max}_{[t, t+\Delta t]}X_t \le M \vert X_{t+\Delta t} = X_t + \Delta X_t)p(\mathrm{max}_{[t+\Delta t, T]} X_t = M \vert X_T = b).
\end{align*}

We integrate to find
\begin{align}
&  P(\mathrm{max}_{[t_1, t_2]}X_t \le M ) \nonumber\\
& =  \int_{-\infty}^M p(\mathrm{max}_{s \in [t+\Delta t, T]} X_s \in dM' \,\vert\, X_{t+\Delta t} = X + \Delta X) dM' \nonumber\\
& = 1-e^{-2(M-a)(M-b)/T}. \label{prob_bound_bb}
\end{align}


\begin{rem}
As $\Delta t \to 0$, $$p(\mathrm{max}_{[t, t+\Delta t]} X_t = M \vert X_{t+\Delta t} = X+\Delta X) \to 0$$ and $$\vert p(\mathrm{max}_{[t, t+\Delta t]}X_t \le M \vert X_{t+\Delta t} = X + \Delta X) -1\vert \to 0$$ like $\mathcal{O}(\frac{k_1}{\Delta t}e^{-k_2/{\Delta t}})$, and so both are dominated as $\Delta t \to 0$ by the remaining factors in the full quotient, which take the form $$k_1 \frac{T-t-\Delta t}{T-t}e^{k_2(\frac{1}{T-t-\Delta t} - \frac{1}{T-t})}.$$
Thus, for the purposes of computation, as $\Delta t \to 0$ we can take the maximum over $[t+\Delta t, T]$ instead of $[t, T]$ for the numerator, simplifying the expression in implementation when $\Delta t$ is very small. In this paper we keep the full expression.
\end{rem}

This gives us an SDE in the first case, as in the algorithm we can select $dX$ at random from this distribution:
\begin{align*}
& p(X_{t+\Delta t} = X + \Delta X\, \vert \, \mathrm{max}_{s \in [t+\Delta t, T]} X_s = M, X_t = b) \\
= & \frac{T-t_0}{T-t_0-\Delta t} \frac{2M-X-\Delta X - b}{2M-X-b} e^{2\frac{(X-M)(b-M) - (X+\Delta X - M)(b-M)}{(T-t_0)(T-t_0-\Delta t)}}\\
& \times \frac{1}{\sqrt{2\pi \sigma^2 \Delta t}}e^{-\frac{-(\Delta X - \frac{b-X}{T-t}\Delta t)^2}{2\sigma^2 \Delta t}}\mathbbm{1}_{\{M \ge X + \Delta X\}}\mathbbm{1}_{\{M\ge b\}}.
\end{align*}

\item \textbf{The process has achieved its maximum already.} In this case we do not require the maximum to be attained on $[t, T]$ (in fact, it almost surely is not), and just require that $X_t$ is bounded from above by $M$ there, so that we have
\begin{multline*}
P(\mathrm{max}_{s\in [t, b]}X_s \le M \vert X_{t+\Delta t} = X + \Delta X)\\
= P(\mathrm{max}_{s \in [t, t+\Delta t]} X_s \le M \vert X_{t+\Delta t} = X_t+ \Delta X_t)\\
 \times P(\mathrm{max}_{s\in [t+\Delta t, T]} X_s \le M \vert  X_T = b)
\end{multline*}
which we can find explicitly from equation \ref{prob_bound_bb}.

\item \textbf{Endpoint.} Until now we have assumed that the $X_t$ had an interval of nonzero length to its right. Finally, after incrementing through all but the last value, we must set $X_T = b$.
\end{enumerate}

The whole procedure follows similarly for minima. However, there are numerical considerations beyond those for the usual bounds and stepsizes in numerical integration. For example, to ensure that truncation error does not hinder the pre-extremum/post-extremum bifurcation in the algorithm, we set a range of acceptable closeness $\varepsilon$ to the extremum when we check whether the extremum has been attained. When incrementing, we should also set a range of $L$ possible increments to select from, weighted by the correct probabilities.

\subsubsection{Rectification}

After approximating a Wiener process meeting the conditions in an unbiased way, a simple rectification of the conditions as follows is justifiable for practical purposes. That is, given a range of acceptable closeness to the maximum, $\varepsilon$, to ensure that the error does not hinder the pre-extremum/post-extremum bifurcation in the algorithm in this subsection, we scale the final process as follows.

Let $\tilde{M}$ be the actual maximum of the numerically generated bridge. First we adjust to ensure the maximum is attained exactly (rather than within $\varepsilon$) by postcomposing with $$x \mapsto a + \frac{M-a}{\tilde{M}-a}(x-a).$$ Denote the new right-hand end point by $\tilde{b}$. Then, to adjust the final value back to $b$, we replace only that section of the bridge after the maximum by its postcomposition with $$x \mapsto M + \frac{M-b}{M-\tilde{b}}(x-M).$$ Then we have traded one error for another: we have slightly adjusted the dynamics, which may be viewed as having a slight effect on $\sigma$, but the more visible endpoints and maximum are held to exactly.

\subsubsection{Implementation}

We present the algorithm below, which may be applied to more general processes.

\begin{algorithm}[H]
\caption{Bridge generation constrained to a given maximum}
\label{alg:bayesian_increments}
\textbf{Input.} $t_0, T, a, b, M,$ dynamical parameters,
$N_{\mathrm{timesteps}}, \delta, \varepsilon, L$
\begin{algorithmic}
\Procedure{DensityMax}{max; $t_1$, $t_2$, $x_1$, $x_2$}
\State density $\leftarrow$ conditional density of bridge maximum max
\State \qquad\qquad\; (for Brownian bridges, given by equation \ref{density_max_bb})
\State \Return density
\EndProcedure
\Procedure{ProbBound}{bound; $t_1$, $t_2$, $x_1$, $x_2$}
\State density $\leftarrow \int_{\mathrm{min}(a, b)}^{\mathrm{bound}} DensityMax(M'; t_1, t_2, x_1, x_2) dM'$
\State \Return density
\EndProcedure
\Procedure{DensityIncrementUnconstrained}{$dx; t, dt, x$}
\State density $\leftarrow$ density of increment determined from unconstrained SDE
\State \qquad\qquad\;(for Brownian bridges, given by equation \ref{sde_bb})
\State \Return density
\EndProcedure
\Procedure{DensityIncrementMax}{$dX; t, dt, t, X, b$, maxAttained}
\State pBoundLeft $\leftarrow$ ProbBound(max; $t, t+dt, X, X+dX$)
\State pBoundRight $\leftarrow$ ProbBound(max; $t+dt, T, X+dX, b$)
\State pMax $\leftarrow$ DensityMax(max; $t, T, X, b$)
\If{not maxAttained}
\State pMaxLeft $\leftarrow$ DensityMax(max; $t, t+dt, X, X+dX$)
\State pMaxRight $\leftarrow$ DensityMax(max; $t+dt, T, X+dX, b$)
\State density $\leftarrow$ DensityIncrementUnconstrained($dX; t, dt, X$)*
\State \qquad\qquad\;\;(pMaxLeft*pBoundRight+pBoundLeft*pMaxRight)/pMax
\Else
\State density $\leftarrow$ DensityIncrementUnconstrained($dX; t, dt, X$)*
\State \qquad\qquad\;\;(pBoundRight+pBoundLeft)/pMax
\EndIf
\State \Return density
\EndProcedure
\Procedure{Main}{Input}
\State $dt \leftarrow (b-a)/N_{\mathrm{timesteps}}$
\State $X_0 \leftarrow a$ \Comment first bridge point
\State maxAttained $\leftarrow$ False
\State $X_{\mathrm{min}} \leftarrow$ $X$ such that
\State \qquad \qquad $\int_{-\infty}^{X_{\mathrm{min}}}$ DensityIncrementUnconstrained($t, t+\Delta t, X, X') dX' < \delta$
\For{$i =1:N_{\mathrm{timesteps}}-2$}
\State $dX_{\mathrm{min}} \leftarrow X_{\mathrm{min}} - X_i$
\State $dX_{\mathrm{max}} \leftarrow M - X_i$
\State $dX \leftarrow$ random variable selected according to
\State \qquad DensityIncrementMax($t, dt, X_i, dX_i, b$, maxAttained) over interval
\State \qquad $[dX_{\mathrm{min}}, dX_{\mathrm{max}}]$ and $L$ steps
\If{$\vert M-X_i\vert < \varepsilon$}
\State maxAttained = True
\EndIf
\State $X_{i+1} \leftarrow X_i + dX_i$
\EndFor
\State $X_{N_{\mathrm{timesteps}}} \leftarrow b$ \Comment last bridge point
\State Optional: rectify the vector $(X_i)$ by scaling as in section 2.2.2
\EndProcedure
\end{algorithmic}
\textbf{Output.} A time vector $X_i$, approximating the bridge with given dynamics over $[t_0, T]$ from $a$ to $b$ constrained to the maximum $M$.
\end{algorithm}

The choice of $X_{\mathrm{min}}$ may be decided case by case. In general, one approach might be to find when the cumulative probability of the minimum (without constraint to a maximum) $$P(\mathrm{min}_{s \in [a, b]}X_s \le X_{\mathrm{min}} \vert X_0 = a, X_T = b) < \delta,$$ which in turn can be found directly from the probBound function by considering the dual case of the process $-X_t$. In the case of the Brownian bridge constrained to a maximum from $a$ to $b$ at a given $t$, we have $$1-L = e^{-2(a-X_{\mathrm{min}})(b-X_{\mathrm{min}})/T}.$$ Solving for $X_{\mathrm{min}}$, we find $$\frac{a+b}{2} - \frac{1}{2} \sqrt{(a-b)^2-2T\ln (1-L)},$$ the only one of the two solutions less than $a, b$. It would also be possible to compute a new $X_{\mathrm{min}}$ at each step over the interval $[X_t, b]$, but this is slightly more computationally expensive and relies on previous simulations rather than taking account of $a$.

Below we show corresponding unrectified and rectified examples of Wiener bridges generated under method 2 (construction by Bayesian increments) over the time interval $[0, 2]$ from 3 to 4 with standard deviation $1$, conditioned on maximum $5$. We take $N_{\mathrm{timesteps}} = 100, \delta = 10^{-3}$ and $\varepsilon = 0.2$, with $L = 10000$ the number of possible increments from which to make the density-weighted selection at each increment.

\begin{figure}[H]
\centering
\includegraphics[width=0.4\textwidth, angle=0]{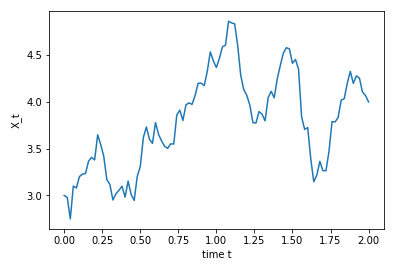} \includegraphics[width=0.4\textwidth, angle=0]{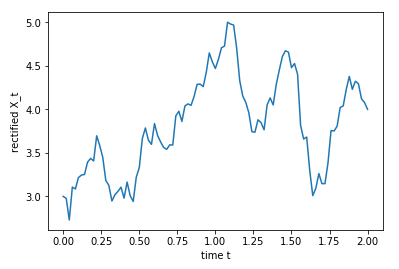}\\
\includegraphics[width=0.4\textwidth, angle=0]{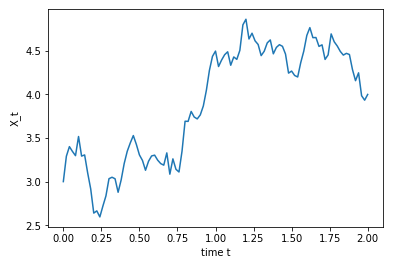} \includegraphics[width=0.4\textwidth, angle=0]{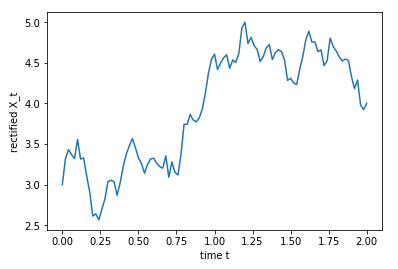}\\
\includegraphics[width=0.4\textwidth, angle=0]{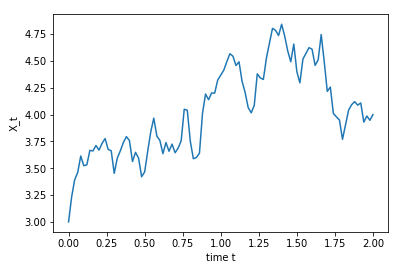} \includegraphics[width=0.4\textwidth, angle=0]{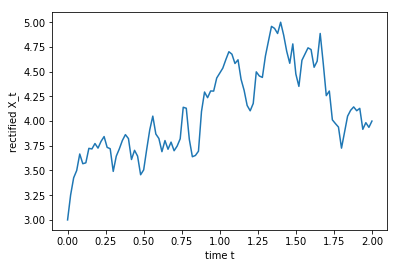}\\
\caption{Plots of unrectified and rectified Wiener bridges generated by method 2 on $[0, 1]$ from $a = 1$ to $b = 3$ conditioned on a maximum of 5, with $\sigma = 1$}
\end{figure}

We note that the results indeed generate bridges subject to the maximal condition, the apparent variance behaves as expected, and the difference between corresponding unrectified and rectified bridges is minor upon inspection.

To test that the numerical generation is reliable, we may verify that the distribution of the argmax also follows the result from \cite{AC}:

\begin{figure}[H]
\centering
\includegraphics[width=0.4\textwidth, angle=0]{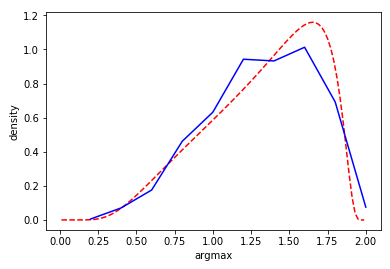}\\
\caption{Density of the location of the maximum of a Brownian bridge from 3 to 4 conditioned on maximum 5, variance 1. Solid (blue): Experimental frequency after 1000 runs as generated by Method 2. Dashed (red): Theoretical density based on equation \ref{density_argmax_wiener}.}
\end{figure}

These visibly agree quite closely. We conclude the fit is good and the simulation is successful for the choices made for these numerical parameters.

We note that for the same parameters, Method 1 (construction via Brownian meanders) took 0.004 seconds to generate the bridge and Method 2 (construction via Bayesian increments) took 10.14 seconds \thanks{using a 2.6GHz Intel Core i5-7300U CPU}. Both methods are linear in the number of timesteps. Clearly the method using meanders is much faster, but the method of Bayesian increments is more generalizable whenever the unconditioned SDE and the density of the extremum are known.

\section{Ornstein-Uhlenbeck bridges conditioned on extrema}


We seek a construction of an OU-process $X$ from $X_{t_0}=a$  to $T$ with volatility $\sigma$ conditioned on $X_T=b$ and $\mathrm{max}_{t_0 \le s \le T} X_s = M$, without bias.
Again, consider $t\in [t_0, T]$ and assume the standard measure.

We note that it is possible to run through algorithm \ref{alg:bayesian_increments} as before, provided we have the density for a given increment in an unconstrained OU bridge, and the density of the maximum of a general OU bridge. The solution of the latter is a difficult problem, and unfortunately one convergent solution we discuss requires multiple nested loops to compute. Even with a more easily computable approximation, both Python and C++ implementations are still impractically slow on a 2.6GHz Intel CPU. Analysis of results using a GPU are pending further research.

Here and elsewhere, we note that the location of an extremum for any diffusion process is almost surely unique (see \cite{KS}).




We find in the unconditioned OU bridge case that
$$\Delta X \sim \mathcal{N}((-\kappa X_t +2\kappa \frac{(b-\mu)e^{\kappa(T+t)}-X_t\kappa e^{2\kappa t}}{e^{2\kappa T}-e^{2\kappa t}})\Delta t, \sigma \sqrt{\Delta t}),$$
that is
\begin{align*}
& p(\mathrm{max}_{s \in [t+\Delta t, T]} X_s = M \,\vert\, X_{t+\Delta t} = X + \Delta X, X_t = b))\\
&= \frac{1}{\sqrt{2\pi \Delta t} \sigma}e^{(\Delta X_t + (\kappa X_t -2\kappa \frac{(b-\mu)e^{\kappa(T+t)}-X_t\kappa e^{2\kappa t}}{e^{2\kappa T}-e^{2\kappa t}})\Delta t)^2/2\sigma^2\Delta t}.
\end{align*}

We find from \cite{CFS} that for a general linear diffusion we have the joint probability density form
\begin{align*}
p_{joint} := &\; p(\mathrm{max}_{t \in [t_0, t_1]} X_t = M,X_t=X_t,\mathrm{argmax}_{t \in [t_0, t_1]} X_t=T)\\
 = & \;n(t, M-a)n(T-t, M-b)S(M)m(db)dt,
 \end{align*}
where $S(y)$ is the scale function, $m(dz)$ is the speed measure for the unconditioned OU process, and $n(s; y)$ is the density of the hitting time of $y$ at $s$ for a process with the same dynamics starting at $0$.
Our $M$ and $b$ are given, so we normalize this by $$\int_a^{\infty}\int_{-\infty}^M p_{joint}(M, b, t) db dM$$ to find the conditional probability density.

For linear diffusions with Markov processes given by $$\mathcal{L} f(x) = \frac{1}{2}\sigma(x)^2 f''(x) + b(x)f'(x) - cf(x),$$ we have
$$S(x) = \int_c^x e^{-B(y)}dy, \quad m(dx) = \frac{2}{\sigma^2}e^{B(x)}dx,$$ where $$B(x) = \int_0^x \frac{2}{\sigma^2(y)}b(y)dy.$$
The Markov generator of the OU process is $$\mathcal{L} f(x) = \frac{1}{2}\sigma^2 f''(x) + \kappa (\mu - x)f'(x).$$ Therefore, given we set the process to start at time 0, we may take $$S(x) = \int_0^x e^{\frac{\theta}{\sigma^2}(y^2-2\mu y)}dy, \quad m(x) = \frac{2}{\sigma^2}e^{\frac{2}{\sigma^2}\theta(\mu x - \frac{1}{2}x^2)}.$$

The hitting time $n_x(s;y)$ does not have a known easily tractable ``analytic'' solution but may be efficiently given as the limit of a sequence given by Lipton and Kaushansky (2017) \cite{LK} based on a method found in Linz \cite{L}, or approximated by other methods detailed further below.

Lipton and Kaushansky proceed by transforming the hitting time problem to a system of equations whose solution depends on the solution to Volterra equations. In particular, they first normalize the OU process to a simpler OU process
$$ dX_t = \kappa (\mu - X_t) dt + \sigma dW_t \quad \leftrightarrow \quad  dX_t = -X_t dt + dW_t,$$
with the same initial value, and then demonstrate that the hitting time $\mathbb{P}_x$-density of a thus normalized OU process starting at $a$ hitting $b$ and evaluated at $t$, is
\begin{align*}
n_x(t;a, b) = & 4e^{-2t} \int_0^{\theta} (-3(\theta - \theta')(2-\theta-\theta') + (a(1-\theta) - b(1-\theta'))^2)\\
& \times (a(1-\theta) - b(1-\theta))\frac{\mathrm{exp}(-\frac{(a(1-\theta) - b(1-\theta'))^2}{(\theta-\theta')(2-\theta-\theta')})(1-\theta')\nu^b(\theta')}{\sqrt{\pi (\theta-\theta')^7(2-\theta-\theta')^7}}d\theta'\\
& + 4e^{-t} \int_0^{\theta} (-(\theta - \theta')(2-\theta-\theta') + (a(1-\theta) - b(1-\theta'))^2)\\
& \times \frac{\mathrm{exp}(-\frac{(a(1-\theta) - b(1-\theta'))^2}{(\theta-\theta')(2-\theta-\theta')})(1-\theta')\nu^b(\theta')}{\sqrt{\pi (\theta-\theta')^5(2-\theta-\theta')^5}}d\theta',
\end{align*}
for $\theta = 1-e^{-t}$, and where $\nu^b(\theta)$ is the solution to a Volterra-type integral equation of the second kind,

$$\nu^b(\theta) - \frac{2b}{\sqrt{\pi}} \int_0^{\theta} \frac{e^{-b^2\frac{\theta-\theta'}{2-\theta-\theta'}}(1-\theta')\nu^b(\theta')}{\sqrt{(2-\theta-\theta')^3(\theta-\theta')}}d\theta'  - 1 = 0.$$

Thus to find the required density, it remains to solve for $\nu^b(\theta)$ above.

We define $$K(\theta, \theta') = -\frac{e^{-b^2\frac{\theta-\theta'}{2-\theta-\theta'}}(1-\theta')\nu^b(\theta')}{\sqrt{(2-\theta-\theta')^3(\theta-\theta')}},$$ the kernel of this integral operator.

Lipton and Kaushansky propose a ``block by block'' method based on a method found in \cite{L}, based on quadratic interpolation, as follows.

First, define the following functions
\begin{align*}
& \alpha(x, y, z) = \frac{z}{2}\int_0^2 \frac{(1-s)(2-s)}{\sqrt{x-y-sz}}ds,\\
& \beta(x, y, z) = z\int_0^2\frac{s(2-s)}{\sqrt{x-y-sz}}ds,\\
& \gamma(x, y, z) = \frac{z}{2}\int_0^2\frac{s(s-1)}{\sqrt{x-y-sz}}ds.
\end{align*}

For faster implementation, these evaluate to
\begin{align*}
& \left\{
	\begin{array}{r}
		\alpha (x, y, z) \\
		\beta  (x, y, z)\\
		\gamma  (x, y, z)
	\end{array}\right\}
= \left\{
	\begin{array}{r}
		-1\\
		2\\
		-1
	\end{array}\right\}[\frac{1}{z}\xi^5 + \frac{1}{3}(-2\frac{x-y}{z^2} + \left\{\begin{array}{r}
		3\\
		2\\
		1
	\end{array}\right\}\frac{1}{z})\xi^3   \\
&\qquad\qquad\qquad\qquad +((\frac{x-y}{z})^2-\left\{\begin{array}{r}
		3\\
		2\\
		1
	\end{array}\right\}\frac{x-y}{z} + \left\{\begin{array}{r}
		2\\
		0\\
		0
	\end{array}\right\})\xi]_{\sqrt{x-y}}^{\sqrt{x-y-2z}}.
\end{align*}

Further, define $$w_{n, i} = \left\{
	\begin{array}{ll}
		(1-\delta_{i, n-1})\alpha(t_n, t_i, h) + (1-\delta_{i,0})\gamma(t_n, t_i-2h, h),  & \mbox{if } i \equiv_2 0 \\
		\beta(t_n, t_i-h, h), & \mbox{if } i \equiv_2 1
	\end{array}
\right.$$

Then for timesteps $0 = t_0 < t_1 < \ldots < t_N = T$, and a desired stepsize $h$, set $F_0 = f(0)$, repeat the following steps for $m=0$ to $N/2$:
\begin{enumerate}
\item For $i = 0, \ldots, 2m$, compute $w_{2m+1, i}, w_{2m+2, i}$.
\item Solve the following system of equations for $[F_{2m+1}, F_{2m+2}]$ and append them to the list of known $F_k$:
\begin{align*}
F_{2m+1} = 1 & + (1-\delta_{m0}) \sum_{i=0}^{2m} w_{2m+1, i}K(t_{2m+1}, t_i)F_i\\
				& + \alpha(t_{2m+1}, t_{2m}, \frac{h}{2})K(t_{2m+1}, t_{2m})F_{2m}\\
				& + \beta (t_{2m+1}, t_{2m}, \frac{h}{2})K(t_{2m+1}, \frac{t_{2m}+t_{2m+1}}{2})\\
				& \qquad \qquad \times (\frac{3}{8}F_{2m} + \frac{3}{4}F_{2m+1} - \frac{1}{8} F_{2m+2})\\
				& + \gamma(t_{2m+1}, t_{2m}, \frac{h}{2})K(t_{2m+1}, t_{2m+1})F_{2m+1},\\
F_{2m+2} = 1 & + (2-\delta_{m0})\sum_{i=0}^{2m} w_{2m+2, i}K(t_{2m+2}, t_i)F_i\\
				& + \alpha(t_{2m+2}, t_{2m}, h)K(t_{2m+2}, t_{2m})F_{2m}\\
				& + \beta(t_{2m+2}, t_{2m}, h)K(t_{2m+2}, t_{2m+1})F_{2m+1}\\
				& + \gamma(t_{2m+2}, t_{2m}, h) K (t_{2m+2}, t_{2m+2})F_{2m+2}.
\end{align*}

Explicitly, at each stage we set
\begin{align*}
&[\beta K]_m = \beta(t_{2m+1}, t_{2m}, \frac{h}{2})K(t_{2m+1}, \frac{t_{2m}+t_{2m+1}}{2})\\
&A_m = -1 + \frac{3}{4}[\beta K]_m + \gamma(t_{2m+1}, t_{2m}, \frac{h}{2})K(t_{2m+1}, t_{2m+1})F_{2m+1},\\
&B_m = -\frac{1}{8}[\beta K]_m\\
&C_m = 1 + (1-\delta_{m0})\sum_{i=0}^{2m} w_{2m+1, i} K(t_{2m+1}, t_i)F_i \\
& \qquad\qquad\qquad\qquad + \alpha(t_{2m+1}, t_{2m}, \frac{h}{2})K(t_{2m+1}, t_{2m})F_{2m} + \frac{3}{8}[\beta K]_m F_{2m},\\
&A'_m = \beta(t_{2m+2}, t_{2m}, h)K(t_{2m+2}, t_{2m+1}),\\
&B'_m = \gamma(t_{2m+1}, t_{2m}, \frac{h}{2})K(t_{2m+2}, t_{2m+2}),\\
&C'_m =  1 + (1-\delta_{m0}) \sum_{i=0}^{2m}w_{2m+1, i} K(t_{2m+1}, t_i)F_i\\
& \qquad\qquad\qquad\qquad + \alpha(t_{2m+2}, t_{2m}, h)K(t_{2m+2}, t_{2m})F_{2m}.
\end{align*}
Then we have
\begin{align*}
&F_{2m+1} = \frac{-C_mB'_m + C'_mB_m}{A_mB'_m - A'_mB_m},\\
&F_{2m+2} = \frac{-A_mC'_m + A'_mC_m}{A_mB'_m - A'_mB_m}.
\end{align*}
\end{enumerate}

Then as $N \to \infty$, $F_{2m+1}, F_{2m} \to \nu^b$ pointwise on $t_i$, $0\le i \le N$.

All functions and operations used to find the incremental probability density are continuous, and thus it follows that algorithm \ref{alg:bayesian_increments} simulates an unbiased OU process with the given bridge and extremum conditions.



For implementation purposes, we set $X_{\mathrm{min}}(\delta)$ by solving for $$P(\mathrm{min}_{s \in [a, b]}X_s \le X_{\mathrm{min}} \vert X_0 = a, X_T = b) < \delta$$ numerically by minimizing $$\vert 1 - P(\mathrm{max}_{s\in [a, b]} \tilde{X}_s \le -X_{\mathrm{min}} \vert X_0 = -a, X_T = -b) - \delta \vert,$$ given by the methods in this subsection, and where the dynamics of $\tilde{X}_t = -X_t$ are given by parameters $-\theta, -\kappa, \sigma$.  

We may cut down on computation time by using an approximation of the hitting time $n_x(s, y)$ found in \cite{LK} from first-order truncation of the governing Volterra equation to an Abel equation and solving via Laplace transforms. This does not formally converge to the correct result, but greatly decreases computation time, and holds for $t\ll1$. Since the values of $t$ in this context are within the range of the increment $[0, \Delta t]$, we make this assumption. We set $$\nu^{\tilde{M}}(\vartheta) = 2e^{\frac{\tilde{M}^2\vartheta}{2}}\mathcal{N}(\tilde{M}\sqrt{\vartheta}),$$ where $\vartheta = 1-e^{-t}$ and $\tilde{M} = \frac{\sqrt{\kappa}}{\sigma}(M-\mu).$

As before, after approximating an OU process meeting the conditions in an unbiased way, applying the same simple rectification in the algorithms may be justifiable for practical purposes, if the maximum is desired to be reached on one of the actual numerical increments. In this case, such rectification trades errors the constraint variables ($a, b, M$) for errors in the dynamical variables $(\mu, \kappa, \sigma)$, in different ways on either side of the actual maximum.

\section{Open-Ended Ornstein-Uhlenbeck processes conditioned on extrema}

The case of open-ended OU processes $X_t$ conditioned on extrema (without loss of generality, a maximum $M$) proceeds mostly as in algorithm \ref{alg:bayesian_increments}, except that we do not specify a last value $b$ after the loop, and pMaxLeft and pBoundLeft are computed quite differently from pMaxRight and pBoundRight. In particular, pMaxLeft and pBoundLeft are computed using the density and distribution formulae for the bridge case as in section 3 (since for these we condition on the endpoints of the interval $[t, t+\Delta t]$), while pMaxRight and pBoundRight are computed using the density and distribution for the open-ended case (since we no longer have the condition $X_T = b$).

The density of the increment of the unconstrained open-ended OU process $P(X_{t+\Delta t} = X + \Delta X )$ is given by $$\frac{1}{\sqrt{2\pi dt}\sigma}e^{-\frac{(dX_t-\kappa(X_t-\mu)dt)^2}{2\sigma^2dt}},$$ and
\begin{align*}
P_{open}(X_{t+\Delta t} &= X + \Delta X \,\vert\, \mathrm{max}_{s \in [t + \Delta t, T]} X_s = M) \\
& = \int_{-\infty}^M P(X_{t+\Delta t} = X + \Delta X \,\vert\, \mathrm{max}_{s \in [t + \Delta t, T]} X_s = M) db',
\end{align*}
where the right-hand side is as in section 3.

$X_{\mathrm{min}}$ may again be solved for numerically, as in section 3. In actual implementation, the last value $X_T$ cannot be given as $b$, and the incremental process assumes an interval of positive length between $t$ and $T$. For this one point we set $X_T = X_{T-\Delta t}$ for simplicity.

\section{Brownian motion with drift and geometric Brown-ian motion conditioned on extrema}


We consider the case of drift and geometric Brownian motion, particularly important in quantitative finanace as it underlies the Black-Scholes model. The case for bridges is in fact trivial, since a Brownian motion with drift constrained to bridge conditions is in fact independent of its drift. To see this, consider a Brownian motion with drift $c$ on $[0, T]$ from $a$ to $b-cT$, and add a drift $ct$. Then by the SDE (equation \ref{sde_bb}) we have $$d(X_t+ct) = \frac{b-X_t}{T-t}dt + cdt + dW_t = \frac{b - (X_t+ct)}{T-t}dt + dW_t,$$ which is exactly the same as the original SDE in $X_t+ct$. Therefore, the problem of constructing a Brownian bridge with drift to a given extremum is already solved in (3). A geometric Brownian bridge $X^{geom}_t$ on $[0, T]$ from $a$ to $b$ with maximum $M$ and drift $c$ (with $a, b, M$ all necessarily positive) is then precisely (and in law) given by $e^{Y_t},$ where $Y_t$ is a Brownian bridge on $[0, T]$ from $\mathrm{ln}\,a$ to $\mathrm{ln}\,b$ with maximum $\mathrm{ln}\,M$.

Our construction of an open-ended Brownian motion $X_t$ with given variance $\sigma^2$ and drift $c$ (and correspondingly an open-ended geometric Brownian motion) constrained to a given maximum $M$ over $[0, T]$ relies again on a Bayesian incremental method. If we assume an initial value $a$, a volatility $\sigma$, drift coefficient $c$, then the transformation $X_t \mapsto (X_t-a)/\sigma$ induces an initial value 0, a drift coefficient $\frac{c}{\sigma}$ and a maximum $\frac{M-a}{\sigma}$. The problem is thus equivalent to simulating a Brownian motion $X_t$ with drift $c$ and volatility 1, constrained to a maximum $M$ over $[0, T]$ with initial value 0, which we assume without loss of generality.

The density of the maximum of such a Brownian motion with drift may be found in \cite{BO}, from an application of Girsanov's theorem:
\begin{equation}
p(\mathrm{max}_{s\in[0, t]} X_t \in dM) = \sqrt{\frac{2}{\pi T}}e^{-(M-cT)^2/2T}-2ce^{2cM}N(-\frac{M+cT}{\sqrt{T}})\mathbbm{1}_{\{M\ge0\}}dM. \label{density_max_drift}
\end{equation}

Integrating and setting the value at infinity to 1, we find that the probability that $X_t$ is bounded by $M$ on $[0, T]$ is \begin{equation}
P(\mathrm{max}_{s\in [0, t]} X_t \le M) = e^{2cM}(\mathcal{N}(\frac{M+cT}{\sqrt{T}})-1) + \mathcal{N}(\frac{M-cT}{\sqrt{T}}).\label{prob_bound_drift}
\end{equation}

The unconditioned density of an increment $\Delta X$ after a given increment $\Delta t$ is given directly from the defining SDE by
\begin{equation}
p(X_{t+\Delta t } =  X+ \Delta X) = \frac{1}{\sqrt{2\pi \Delta t}\sigma}e^{\frac{-(\Delta X-c\Delta t)^2}{2\sigma^2\Delta t}}. \label{density_increment_drift}
\end{equation}

As in section 4, we follow algorithm \ref{alg:bayesian_increments} with the exception that pMaxLeft and pBoundLeft are computed as for the bridge process over $[t, t+\Delta t]$ from $X$ to $X+\Delta X$ (recalling that these are independent of $c$ and are the same as in section 2.2), and pMaxRight and pBoundRight are given for the open-ended  case. That is, we define pBoundLeft from equation \ref{prob_bound_bb}, pBoundRight from equation \ref{prob_bound_drift}, pMaxLeft from equation \ref{density_max_bb}, pMaxRight by equation \ref{density_max_drift}, and DensityIncrementUnconstrained from equation \ref{density_increment_drift}.

In our implementation, we find $X_{\mathrm{min}}$ numerically by solving for $$P(\mathrm{min}_{s\in [t_0, T]}X_s \le X_{\mathrm{min}})< \delta.$$ We use a Brent solver to minimize $$\vert  e^{2cX_{\mathrm{min}}}(1-\mathcal{N}(\frac{-X_{\mathrm{min}}-cT}{\sqrt{T}})) + (1-\mathcal{N}(\frac{-X_{\mathrm{min}}+cT}{\sqrt{T}})) - \delta \vert.$$ As in section 4, the last value $X_T$ cannot be given by a predefined value $b$, and the incremental process assumes an interval of positive length between $t$ and $T$. For this one point, we assume for simplicity that $X_T = X_{T-\Delta t}$.

We used our implementation to simulate Wiener processes over $[0, 2]$ starting at 3 with drift $c = 1$ and volatility $\sigma = 2$. Numerically, we set $N_{\mathrm{timesteps}} = 100$, $\delta = 10^{-3}$, $\varepsilon = 0.1$, and $L = 1000$. We show examples below.

\begin{figure}[H]
\centering
\includegraphics[width=0.4\textwidth, angle=0]{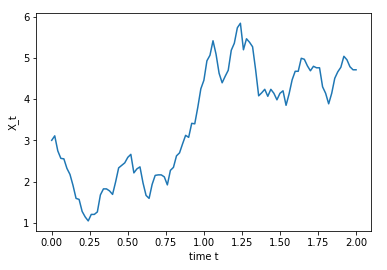} \includegraphics[width=0.4\textwidth, angle=0]{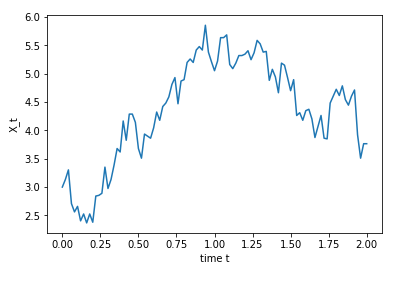}\\
\includegraphics[width=0.4\textwidth, angle=0]{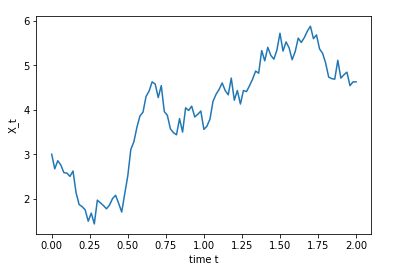} \includegraphics[width=0.4\textwidth, angle=0]{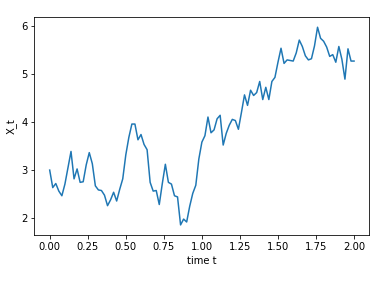}\\
\caption{Plots of Wiener processes with drift $[0, 1]$ with initial value 3, conditioned on a maximum of 6, with drift coefficient 1 and $\sigma = 2$}
\end{figure}

Note that we did not rectify the maximum value. The plots show expected behavior. Mean computation time was 21.69 seconds despite the decrease in $L$ from 10000 to 1000. Setting $L$ to 10000, the mean computation time was 169.47 seconds.\footnote {using a 2.6GHz Intel Core i5-7300U CPU} The density of the location of the maximum (unconstrained on the maximum value) is found in \cite{KS} to be
\begin{equation}
p(\mathrm{argmax}_{s\in[0, T]} X_s = \theta) = 2(\frac{e^{-c^2\theta/2}}{\sqrt{2\pi\theta}} + c\mathcal{N}(c\sqrt{\theta}))(\frac{e^{-c^2(T-\theta)/2}}{\sqrt{2\pi(T-\theta)}} - c\mathcal{N}(-c\sqrt{T-\theta})).\label{density_argmax_drift}\end{equation} This density is the same for $c/\sigma$ in our case.

We see similar distributions for 10 bins, after randomly generating the maximum according to equation \ref{density_max_drift}. We plot both the experimental and theoretical densities for 10 bins below.
\begin{figure}[H]
\centering
\includegraphics[width=0.4\textwidth, angle=0]{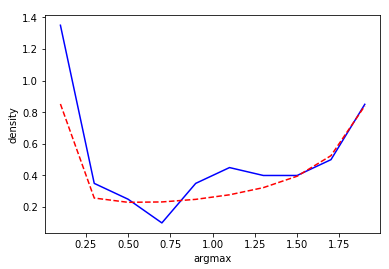}\\
\caption{Density of the location of the maximum of a Wiener process constrained to a maximum of 6 over $[0, 2]$, with drift coefficient 1, volatility 2 and initial value 3. Solid (blue): Experimental frequency after 100 runs. Dashed (red): Theoretical density based on equation \ref{density_argmax_drift}}
\end{figure}
These appear to agree fairly closely. We conclude that the fit is good and the simulation is successful for the choices made for these numerical parameters.

Finally, we exponentiate to illustrate the case for geometric Brownian motion, which we take over $[0, 2]$, starting at 3, conditioned on a maximum of 6. For the bridge case, we simulate the logarithm with the method of 2.1, and assume a value of 4 at $t = 2$, volatility 2 and exponential drift coefficient 1. We simulate 100 timesteps.

\begin{figure}[H]
\centering
\includegraphics[width=0.4\textwidth, angle=0]{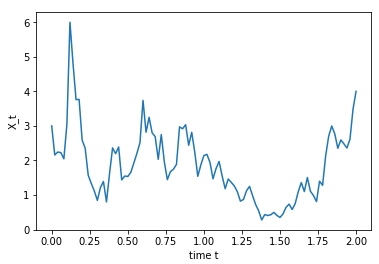} \includegraphics[width=0.4\textwidth, angle=0]{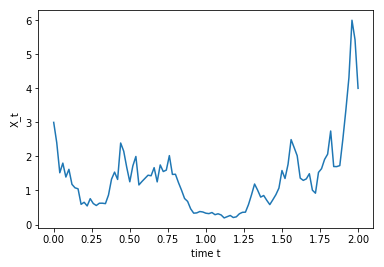}\\
\caption{Plots of Geometric Brownian bridges over  $[0, 1]$ from 3 to 4, conditioned on a maximum of 6, with exponential drift coefficient 1 and volatility $\sigma = 2$}
\end{figure}

For the open-ended case we simulate 100 timesteps with $\delta = 10^{-3}, \varepsilon = 0.1$ and $L = 1000$. First we take drift coefficient 0.2, with volatility 0.1 and 1 respectively:
\begin{figure}[H]
\centering
\includegraphics[width=0.27\textwidth, angle=0]{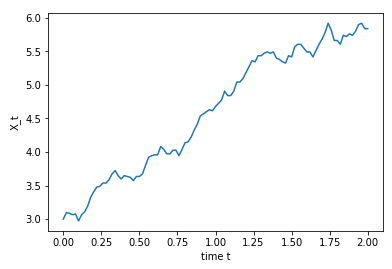} \includegraphics[width=0.27\textwidth, angle=0]{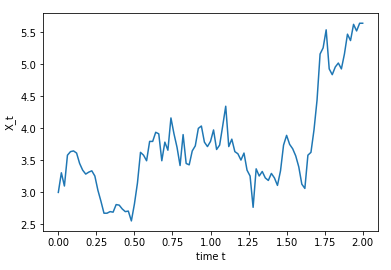} \includegraphics[width=0.27\textwidth, angle=0]{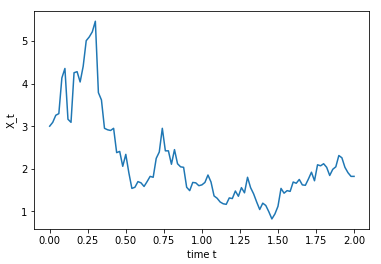}  \\
\caption{Plots of Geometric Brownian motions over $[0, 1]$ with initial value 3, conditioned on a maximum of 6, with exponential drift coefficient 0.2 and from left to right taking $\sigma = 0.1, \sigma = 0.5, \sigma = 1$}
\end{figure}
The simulation meets the maximum condition and the behavior of the simulation as volatility increases as expected.

Increasing the drift coefficient to 0.5, retaining volatility over the same range:
\begin{figure}[H]
\centering
\includegraphics[width=0.27\textwidth, angle=0]{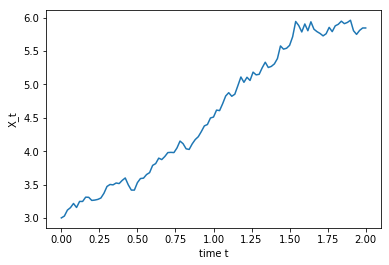} \includegraphics[width=0.27\textwidth, angle=0]{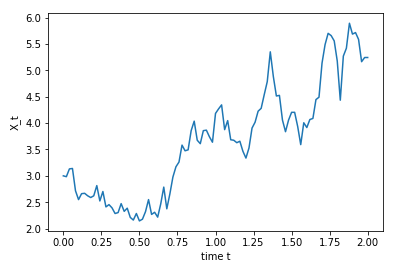} \includegraphics[width=0.27\textwidth, angle=0]{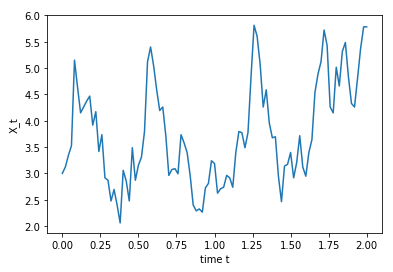}  \\
\caption{Plots of Geometric Brownian motions over $[0, 1]$ with initial value 3, conditioned on a maximum of 6, with exponential drift coefficient 0.5 and from left to right taking $\sigma = 0.1, \sigma = 0.5, \sigma = 1$}
\end{figure}

This construction lends itself to several applications. If the underlying of an option is known to have attained a given maximum, but the current payoff of the option is not known, its value can be estimated by Monte Carlo simulation over geometric Brownian bridge (i.e., Black-Scholes) scenarios conditioned on a given extremum. Barrier option pricing in particular is sensitive to this condition. We will save these applications for a later paper.




\begin{center}
\textbf{DECLARATION OF INTEREST}
\end{center}

The authors alone are responsible for the content and writing of this paper. The information and views expressed by the authors are their own and not necessarily those of Ernst \& Young LLP or other member firms of the global EY organization.

\end{document}